\documentclass[
showkeys,	
reprint,	
amsmath,	
amssymb,	
amsfonts,	
aps,		
prl%
groupedaddress,
superscriptaddress,
a4paper,%
nofootinbib,
eprint,		
final,		
preprintnumbers, 
]{revtex4-2}

\usepackage[tbtags]{mathtools}
\usepackage{amsmath,amsfonts,mathdots,amssymb,yfonts,calc}
\usepackage{graphicx,graphics,xcolor}
\usepackage{subfigure}
\usepackage{amsthm}

\newtheorem{theorem}{Theorem}

\usepackage{natbib}
\usepackage{booktabs}
\usepackage{siunitx}
\usepackage{float}
\usepackage{tabularray}
\usepackage{multirow}
\usepackage{physics}
\usepackage{dcolumn}
\usepackage[utf8]{inputenc}  
\usepackage[T1]{fontenc} 
\usepackage{natbib}
\usepackage{booktabs}
\usepackage{hyperref}
\hypersetup{
    colorlinks=true,
    linkcolor=blue,
    filecolor=magenta,      
    urlcolor=blue!80,
    citecolor=cyan,
    linktocpage=true,
    breaklinks=false,
}

\begin{document}
\title{Geometry Induced Adaptive Dissipation in Non Linear Contact Hamiltonian Systems: Theory and Application to Duffing Oscillator}
\author{Vinesh Vijayan}
\email[]{vinesh.physics@rathinam.in}
\affiliation{Department of Science and Humanities, Rathinam Technical Campus, Coimbatore, India -641021}

\author{R Sathishkumar}
\affiliation{Department of Science and Humanities, Rathinam Technical Campus, Coimbatore, India -641021}
\author{D Kaleeswaran}
\affiliation{Department of computer science and engineering, Rathinam Technical Campus, Coimbatore, India -641021}

\date{\today}

\begin{abstract}
We develop a generalized contact Hamiltonian framework by extending canonical contact Hamiltonian mechanics to nonlinear dissipative systems through the replacement of the classical linear contact potential with a smooth nonlinear contact potential. The proposed formulation establishes a generalized energy dissipation law together with a structural characterization of admissible contact-induced damping, introducing effective contact dissipation as an intrinsic geometric measure of adaptive dissipation. As an application, a generalized contact Duffing oscillator is derived, in which dissipation emerges intrinsically from the contact geometry rather than being introduced phenomenologically. Numerical investigations of the generalized contact Duffing oscillator demonstrate that nonlinear contact geometry governs the effective dissipation and produces significant changes in the phase-space structure, energy decay, dynamical stability, and long-term nonlinear behavior. The proposed theory therefore provides a systematic geometric framework for constructing and analyzing nonlinear dissipative Hamiltonian systems within contact Hamiltonian mechanics.
\end{abstract}

\maketitle
\pagestyle{plain}

\section{\label{s1}Introduction}
In nonlinear dynamical systems arising in physics, mechanics, engineering, and thermodynamics, dissipation is ubiquitous. Friction, damping, irreversible energy transfer, and interactions with the surroundings strongly influence the qualitative behavior of real systems by governing their stability, bifurcation, and asymptotic dynamics\cite{Morrison1986}. However, the symplectic formulation of classical Hamiltonian mechanics is inherently restricted to reversible dynamics and cannot naturally accommodate dissipative processes\cite{Bravetti2017AnnPhys}. This limitation has motivated the development of geometric formulations capable of incorporating dissipation and has attracted considerable attention over the past few decades\cite{Bravetti2017Entropy,deLeonLainz2019}.

One of the most promising extensions of classical Hamiltonian mechanics is contact Hamiltonian mechanics\cite{Eberard2007,Bravetti2017AnnPhys}. By introducing an additional coordinate, the \emph{contact coordinate}, the symplectic phase space is extended to a $(2n+1)$-dimensional contact manifold, naturally incorporating irreversible processes while preserving the Hamiltonian structure of the dynamics\cite{Eberard2007,Bravetti2017AnnPhys,deLeonLainz2019}. This framework provides a unified description of both conservative and dissipative systems, with dissipation arising intrinsically from the contact geometry rather than being introduced phenomenologically\cite{Bravetti2017AnnPhys,Bravetti2015}. Building on the foundational works of Bravetti, Cruz, and Tapias, modern contact Hamiltonian mechanics extends the fundamental concepts of symplectic Hamiltonian dynamics, including Hamiltonian vector fields, canonical equations, and invariant measures, to the contact setting\cite{Bravetti2017AnnPhys,deLeonLainz2019}.

Contact Hamiltonian mechanics has witnessed extensive development in recent years and has found applications in a wide range of areas, including irreversible thermodynamics, statistical mechanics\cite{Bravetti2015}, Hamilton--Jacobi theory\cite{deLeonLainz2019}, variational principles\cite{Wang2017}, contact dynamical systems\cite{Bravetti2017AnnPhys}, field theories\cite{Gaset2020}, and geometric numerical methods\cite{Vermeeren2019}. This mature mathematical framework provides a natural foundation for the analysis of open and dissipative systems. Its versatility has further enabled unified Lagrangian--Hamiltonian formulations, contact field theories, and the treatment of systems with nonholonomic constraints\cite{deLeon2020,deLeon2021}.

Despite these advances, an important limitation remains in existing contact Hamiltonian formulations. Nearly all current developments employ a linear contact potential of the form $\Gamma(s)=\lambda s$, which produces constant dissipation\cite{Bravetti2017AnnPhys, Wang2017, Gaset2020, deLeon2023HJ, PerezAlvarez2024}. Although this formulation successfully extends Hamiltonian mechanics to dissipative systems, it represents only a restricted class of dissipative dynamics. In many physical systems, dissipation evolves during the motion and depends on the system state, giving rise to adaptive damping, bifurcations, and nonlinear energy transfer. The influence of nonlinear contact geometry on such adaptive dissipation and nonlinear oscillatory dynamics remains largely unexplored\cite{Ravindra1994, Baltanas2001, Zaitsev2012, AlHababi2020, Perkins2024}. Consequently, a generalized contact Hamiltonian framework capable of generating nonlinear dissipative dynamics directly from the contact geometry has yet to be developed systematically.

In this work, we develop a generalized contact Hamiltonian framework by replacing the classical linear contact potential with a smooth nonlinear contact potential. The proposed formulation preserves the canonical contact Hamiltonian structure while allowing the dissipation mechanism to emerge naturally from the underlying contact geometry. Within this framework, we derive a generalized energy dissipation law, establish a structural characterization of contact-induced damping, and introduce the concept of effective contact dissipation as a geometric measure of the state-dependent dissipation generated by the nonlinear contact potential.

As an application, we formulate a generalized contact Duffing oscillator, using the Duffing system as a prototype for nonlinear dynamics. Unlike conventional approaches, in which dissipation is introduced phenomenologically, the proposed model generates dissipation intrinsically through the underlying contact geometry. Numerical investigations demonstrate how the nonlinear contact potential influences the effective dissipation, phase-space dynamics, energy decay, bifurcation structure, and dynamical stability. The proposed framework extends our previous work and provides a geometric approach for investigating nonlinear dissipative oscillations within contact Hamiltonian mechanics\cite{Vijayan2026}.

\section{\label{s2}Generalized Contact Hamiltonian Framework}
\subsection{Contact Hamiltonian Mechanics}

Let $Q$ be an $n$-dimensional configuration manifold with local coordinates $q=(q^1,\ldots,q^n)$, and let $\mathcal{M}=T^{*}Q\times\mathbb{R}$ be the $(2n+1)$-dimensional contact manifold equipped with the canonical contact one-form
\begin{equation}
\eta=ds-p_i\,dq^i.
\label{E1}
\end{equation}

The Reeb vector field $R$ is uniquely defined by
\begin{equation}
i_R\eta=1,\qquad i_Rd\eta=0,
\label{E2}
\end{equation}
which gives
\[
R=\frac{\partial}{\partial s}.
\]

For a smooth contact Hamiltonian $H:\mathcal{M}\rightarrow\mathbb{R}$, the contact Hamiltonian vector field $X_H$ satisfies
\begin{equation}
i_{X_H}d\eta=dH-(R(H))\eta,
\label{E3}
\end{equation}
together with
\begin{equation}
\eta(X_H)=-H.
\label{E4}
\end{equation}

Equations~(\ref{E3})--(\ref{E4}) yield the canonical contact Hamiltonian equations
\begin{align}
\dot q^i &= \frac{\partial H}{\partial p_i}, \label{E5}\\
\dot p_i &= -\frac{\partial H}{\partial q^i}-p_i\frac{\partial H}{\partial s}, \label{E6}\\
\dot s &= p_i\frac{\partial H}{\partial p_i}-H. \label{E7}
\end{align}

Unlike symplectic Hamiltonian flows, contact Hamiltonian flows are non-conservative. The additional contact coordinate $s$ provides a geometric mechanism for irreversible energy exchange while preserving the Hamiltonian framework.

For linear dissipation, the contact Hamiltonian is
\begin{equation}
H(q,p,s)=H_0(q,p)+\lambda s,
\label{E8}
\end{equation}
where $H_0$ is the conservative Hamiltonian and $\lambda>0$ is the dissipation coefficient. Substituting (\ref{E8}) into (\ref{E6}) gives
\begin{equation}
\dot p_i=-\frac{\partial H_0}{\partial q^i}-\lambda p_i,
\label{E9}
\end{equation}
which recovers viscous damping. However, this formulation is restricted to constant dissipation and cannot represent nonlinear or state-dependent damping.

\subsection{Generalized Nonlinear Contact Hamiltonian}

To generalize the dissipation mechanism, consider
\begin{equation}
H(q,p,s)=H_0(q,p)+\Gamma(s),
\label{E10}
\end{equation}
where $H_0:T^{*}Q\rightarrow\mathbb{R}$ is the conservative Hamiltonian and $\Gamma\in C^2(\mathbb{R})$ is an arbitrary nonlinear contact potential. The classical model is recovered by choosing $\Gamma(s)=\lambda s$.

The derivative with respect to the contact coordinate is
\begin{equation}
\frac{\partial H}{\partial s}=\Gamma'(s),
\label{E11}
\end{equation}
which motivates the effective dissipation function
\begin{equation}
D(s):=\Gamma'(s).
\label{E12}
\end{equation}

The generalized contact Hamiltonian equations become
\begin{align}
\dot q^i &= \frac{\partial H_0}{\partial p_i}, \label{E13}\\
\dot p_i &= -\frac{\partial H_0}{\partial q^i}-D(s)p_i, \label{E14}\\
\dot s &= \Delta(H_0)-H_0-\Gamma(s), \label{E15}
\end{align}
where the Euler homogeneous operator is
\begin{equation}
\Delta(H_0)=p_i\frac{\partial H_0}{\partial p_i}.
\label{E16}
\end{equation}

Equations~(\ref{E13})--(\ref{E15}) preserve the canonical contact structure while extending the admissible dissipation law from constant viscous damping to arbitrary nonlinear contact-induced dissipation. The entire dissipative mechanism is encoded by the contact potential through $D(s)$; different choices of $\Gamma(s)$ generate different nonlinear dissipation laws without modifying the underlying contact geometry.

For mechanical systems,
\begin{equation}
H_0(q,p)=T(p)+V(q),
\label{E17}
\end{equation}
where
\[
T(p)=\frac12 p^{T}M^{-1}p.
\]
The governing equations reduce to
\begin{align}
\dot q &= M^{-1}p, \label{E18}\\
\dot p &= -\nabla V(q)-D(s)p, \label{E19}\\
\dot s &= p^{T}M^{-1}p-H_0-\Gamma(s). \label{E20}
\end{align}

This formulation provides a unified framework for nonlinear contact-induced dissipation, including state-dependent damping, while preserving the canonical contact geometry.

\paragraph{Energy Dissipation Law}

Let $(q(t),p(t),s(t))$ be a continuously differentiable solution of the generalized contact Hamiltonian system. Then the conservative Hamiltonian satisfies
\begin{equation}
\frac{dH_0}{dt}=-D(s)\Delta(H_0),
\label{E21}
\end{equation}
where $D(s)=\Gamma'(s)$.

\textbf{Proof.}
Using (\ref{E13}) and (\ref{E14}),
\begin{equation}
\begin{aligned}
\frac{dH_0}{dt}
&=\frac{\partial H_0}{\partial q^i}\dot q^i
+\frac{\partial H_0}{\partial p_i}\dot p_i \\
&=-D(s)\,p_i\frac{\partial H_0}{\partial p_i}
=-D(s)\Delta(H_0).
\end{aligned}
\label{E22}
\end{equation}

For mechanical systems,
\begin{equation}
\Delta(H_0)=2T,
\label{E23}
\end{equation}
and hence
\begin{equation}
\frac{dH_0}{dt}=-2D(s)T.
\label{E24}
\end{equation}

Equation~(\ref{E24}) shows that the mechanical energy dissipation rate is governed entirely by the effective contact dissipation function $D(s)$. Consequently, different choices of $\Gamma(s)$ produce different energy decay mechanisms while preserving the underlying contact Hamiltonian structure.

Numerical simulations of the generalized contact Duffing oscillator, Eqs.~(\ref{E39})--(\ref{E41}), are performed to verify the generalized energy dissipation law. Figure~\ref{F2} compares the analytical prediction given by Theorem~2.1 with the numerically computed time derivative of the conservative mechanical energy. The excellent agreement between the analytical and numerical results validates the proposed energy dissipation law and confirms the consistency of the generalized contact Hamiltonian formulation.

\begin{figure}[t]
\centering
\includegraphics[width=\linewidth]{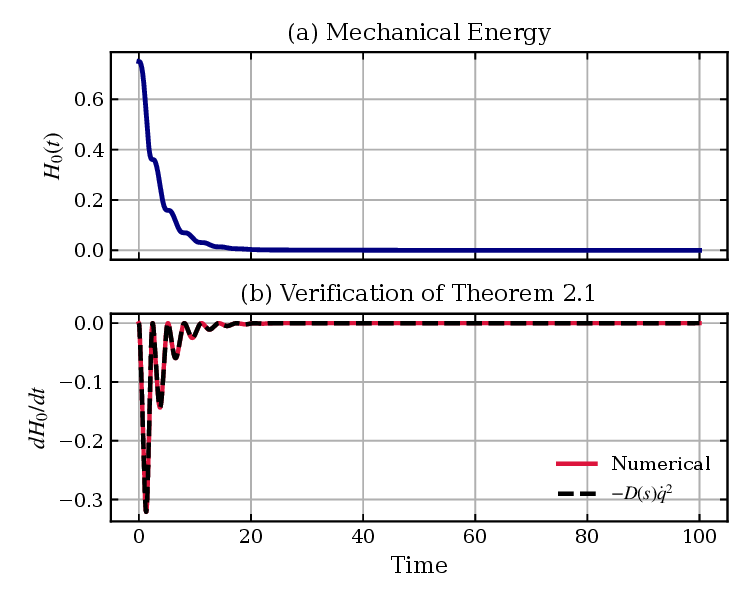}
\caption{Numerical verification of the generalized energy dissipation law (Theorem~2.1). (a) Temporal evolution of the mechanical energy $H_0(t)$ for the generalized contact Duffing oscillator, showing monotonic energy decay due to intrinsic geometric dissipation. (b) Comparison between the numerical energy rate $\mathrm{d}H_0/\mathrm{d}t$ and the theoretical prediction $-D(s)\dot{q}^{\,2}$. The excellent agreement validates the generalized energy dissipation law derived from the proposed nonlinear contact Hamiltonian framework.}
\label{F2}
\end{figure}

Figure~\ref{F2}(a) shows the monotonic decay of the conservative mechanical energy due to the intrinsic geometric dissipation generated by the contact Hamiltonian. Figure~\ref{F2}(b) compares the theoretical prediction with the numerically computed energy dissipation rate. The excellent agreement between the two confirms the validity of Theorem~2.1 and verifies that the proposed generalized contact Hamiltonian accurately captures the energy transfer induced by the nonlinear contact geometry.

\section{\label{s3}Nonlinear Contact Dissipation}

This section examines the mathematical consequences of the generalized contact Hamiltonian by characterizing the evolution of the effective contact dissipation and demonstrating that a broad class of nonlinear damping laws cannot be generated within the canonical linear contact framework.

\subsection{Effective Contact Dissipation}

In the generalized contact formulation, the dissipative force is determined by the local geometry of the contact potential. The momentum equation,
\begin{equation}
\dot{p}_i=-\frac{\partial H_0}{\partial q^i}-D(s)p_i,
\label{E25}
\end{equation}
with
\begin{equation}
D(s)=\Gamma'(s),
\label{E26}
\end{equation}
identifies $D(s)$ as the \emph{effective contact dissipation}. Unlike conventional damping coefficients prescribed externally, $D(s)$ is generated intrinsically by the contact geometry through the contact potential. It represents the instantaneous rate at which the contact structure extracts mechanical energy from the conservative subsystem.

For the quadratic contact potential,
\begin{equation}
\Gamma(s)=\lambda s+\alpha s^2,
\label{E27}
\end{equation}
the corresponding effective dissipation is
\begin{equation}
D(s)=\lambda+2\alpha s.
\label{E28}
\end{equation}

Similarly, for the saturating contact potential,
\begin{equation}
\Gamma(s)=\lambda s+\alpha\tanh(s),
\label{E29}
\end{equation}
the effective dissipation becomes
\begin{equation}
D(s)=\lambda+\alpha\,\operatorname{sech}^2(s).
\label{E30}
\end{equation}

Equations~(\ref{E27})--(\ref{E30}) illustrate that different choices of the contact potential generate distinct nonlinear dissipation laws while preserving the canonical contact structure. Consequently, the generalized formulation admits a broad class of nonlinear damping mechanisms in which the dissipation is an intrinsic geometric quantity determined entirely by the local properties of the contact potential.
\subsection{ Evolution of the Contact Dissipation}
Differentiating (\ref{E26}) with respect to time gives
\begin{equation}
\frac{dD}{dt}=\Gamma''(s)\frac{ds}{dt},
\label{E31}
\end{equation}
and substituting the contact evolution equation (\ref{E15}) yields
\begin{equation}
\dot{D}
=\Gamma''(s)\left[\Delta(H_0)-H_0-\Gamma(s)\right].
\label{E32}
\end{equation}

Equation~(\ref{E32}) governs the evolution of the effective contact dissipation. The dynamics of $D(s)$ are determined solely by the curvature of the contact potential through $\Gamma''(s)$. Hence, different contact potentials generate distinct dissipation dynamics. In particular, $\Gamma''(s)=0$ implies $\dot{D}=0$, recovering the constant dissipation of the classical linear contact Hamiltonian. Therefore, the dissipation coefficient is itself a dynamical variable: the contact coordinate not only evolves with the system but also regulates the dissipation through the geometry of the contact potential.

\begin{theorem}[Structural theorem]
\label{thm:structural}
Let
\[
H(q,p,s)=H_0(q,p)+\Gamma(s),
\]
where $H_0\in C^2(T^*Q)$ and $\Gamma\in C^2(\mathbb{R})$. Then every dissipative force generated by the generalized contact Hamiltonian is necessarily of the form
\begin{equation}
F_d=-D(s)p,
\label{E33}
\end{equation}
where
\[
D(s)=\Gamma'(s).
\]
Consequently, the damping coefficient depends exclusively on the contact coordinate $s$ and cannot depend explicitly on the mechanical variables $q$ or $p$.
\end{theorem}

\begin{proof}
From the generalized contact Hamiltonian (\ref{E10}),
\[
H(q,p,s)=H_0(q,p)+\Gamma(s),
\]
the dependence on the contact coordinate is entirely through the function $\Gamma(s)$. Hence,
\[
D(s)=\frac{\partial H}{\partial s}=\Gamma'(s),
\]
which is a function of $s$ alone. Therefore,
\[
\frac{\partial D}{\partial q^i}=0,
\qquad
\frac{\partial D}{\partial p_i}=0,
\]
for all $i$.

Substituting $D(s)$ into the momentum equation (\ref{E25}) gives
\[
\dot p_i=-\frac{\partial H_0}{\partial q^i}-D(s)p_i,
\]
from which the dissipative force is
\[
F_d=-D(s)p.
\]
Thus every dissipative force arising from the generalized contact Hamiltonian possesses the structure (\ref{E33}), with a damping coefficient depending solely on the contact coordinate.

Conversely, suppose a dissipative force is of the form
\[
F_d=-\widetilde D(q,p,s)\,p,
\]
where $\widetilde D$ depends explicitly on either $q$ or $p$. Then
\[
\frac{\partial \widetilde D}{\partial q^i}\neq0
\quad\text{or}\quad
\frac{\partial \widetilde D}{\partial p_i}\neq0
\]
for some $i$. Such a dependence cannot be represented as $\Gamma'(s)$ for any $\Gamma\in C^2(\mathbb{R})$, since $\Gamma'(s)$ is a function of the contact coordinate alone. This contradicts the defining structure of the generalized contact Hamiltonian (\ref{E10}). Therefore, no canonical contact Hamiltonian of the form (\ref{E10}) can generate damping coefficients with explicit dependence on the mechanical variables $q$ or $p$.
\end{proof}

A classical example is the Van der Pol oscillator,
\begin{equation}
\ddot{q}+\mu(q^2-1)\dot{q}+q=0,
\label{E34}
\end{equation}
whose nonlinear damping term depends explicitly on the configuration variable $q$. Consequently, it lies outside the class of dissipative systems characterized by Theorem~\ref{thm:structural}. The same conclusion applies to a broad class of nonlinear oscillators with state-, velocity-, or amplitude-dependent damping coefficients.

The proposed formulation is not intended to reproduce such damping laws directly. Instead, it establishes a geometric framework in which dissipation is generated intrinsically through the contact coordinate. Rather than prescribing the damping phenomenologically, nonlinear dissipation emerges naturally from the geometry of the contact Hamiltonian via the contact potential $\Gamma(s)$, thereby replacing an externally specified damping law with a geometrically generated, dynamically evolving dissipation mechanism. This geometric viewpoint considerably enlarges the class of admissible contact Hamiltonian models while preserving the canonical contact structure.
\section{\label{s4}Generalized Contact Duffing Oscillator}

This section derives a generalized contact Duffing oscillator in which dissipation is generated intrinsically through the contact potential. Consequently, the effective damping becomes a geometric quantity determined entirely by the contact Hamiltonian.

\subsection{Conservative and Generalized Duffing Oscillator}

Consider the conservative Duffing oscillator with generalized coordinate $q$ and momentum $p$. Its Hamiltonian is
\begin{equation}
H_0(q,p)=\frac{1}{2}p^2+\frac{1}{2}\omega_0^2q^2+\frac{1}{4}\beta q^4,
\label{E35}
\end{equation}
where $\omega_0$ is the natural frequency and $\beta$ is the nonlinear stiffness coefficient. The corresponding Hamiltonian equations are
\begin{align}
\dot q &= p,
\label{E36}\\
\dot p &= -\omega_0^2q-\beta q^3,
\label{E37}
\end{align}
which combine to give the classical Duffing oscillator.

Introducing the nonlinear contact potential,
\begin{equation}
H(q,p,s)=H_0(q,p)+\Gamma(s),
\label{E38}
\end{equation}
the generalized contact Hamiltonian equations become
\begin{align}
\dot q &= p,
\label{E39}\\
\dot p &= -\omega_0^2q-\beta q^3-D(s)p,
\label{E40}\\
\dot s &= \frac12p^2-\frac12\omega_0^2q^2-\frac14\beta q^4-\Gamma(s),
\label{E41}
\end{align}
where $D(s)=\Gamma'(s)$.

Eliminating the momentum variable yields the generalized contact Duffing equation,
\begin{equation}
\ddot q+D(s)\dot q+\omega_0^2q+\beta q^3=0,
\label{E42}
\end{equation}
with the contact evolution governed by (\ref{E41}).

For the linear contact potential $\Gamma(s)=\lambda s$, one has $D(s)=\lambda$, and (\ref{E42}) reduces to the classical contact Duffing oscillator. More generally, the contact potentials introduced in (\ref{E27})--(\ref{E30}) generate the corresponding nonlinear dissipation laws without modifying the conservative Hamiltonian. Consequently, the damping coefficient evolves together with the contact coordinate as an intrinsic dynamical variable rather than a prescribed parameter.

Furthermore, the conservative energy satisfies the dissipation law (\ref{E21}), implying that the nonlinear contact potential continuously regulates the rate of mechanical energy dissipation through the effective contact dissipation $D(s)$. Thus, the generalized contact Duffing oscillator contains the classical contact Duffing model as the linear limiting case while providing a unified geometric framework for nonlinear contact-induced dissipation.

\begin{figure}[t]
\centering
\includegraphics[width=\linewidth]{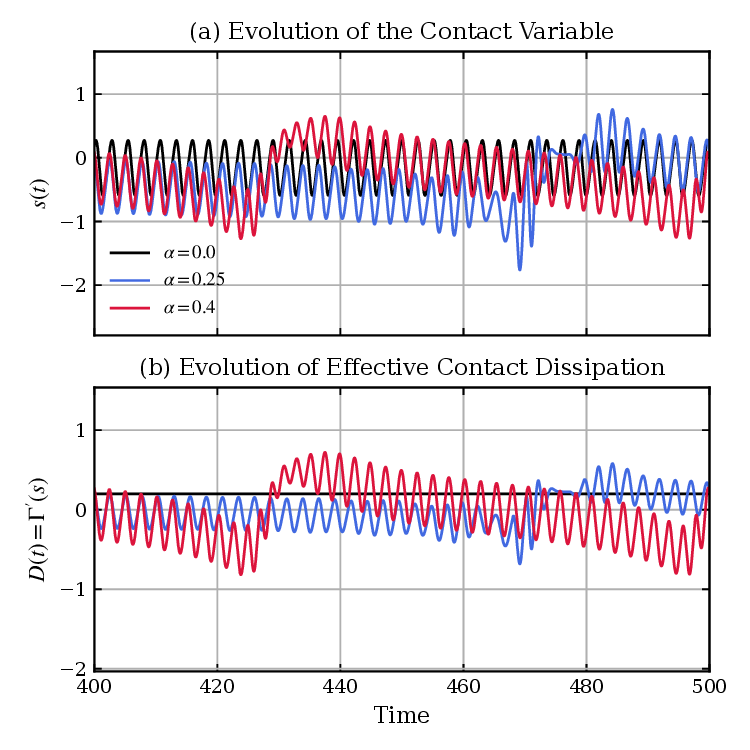}
\caption{Evolution of the contact variable $s(t)$ and the corresponding effective contact dissipation $D(s)=\Gamma'(s)$ for the generalized contact Duffing oscillator with the nonlinear contact potential $\Gamma(s)=\lambda s+\alpha s^2$. The nonlinear contact geometry generates a time-dependent effective dissipation that evolves naturally with the contact variable, demonstrating the adaptive geometric dissipation mechanism inherent in the proposed formulation. {\bf $\alpha = 0$: Constant dissipation (classical contact Hamiltonian)
$\alpha > 0$: Adaptive geometric dissipation (proposed framework)}}
\label{F1}
\end{figure}
Figure~\ref{F1} shows the time evolution of the contact variable and the corresponding effective contact dissipation for different values of the nonlinear parameter $\alpha$ using the quadratic contact potential. For $\alpha=0$, the contact potential is linear and the effective dissipation remains constant, even though the contact variable exhibits oscillatory behavior (black curve). As $\alpha$ increases, the contact variable undergoes larger oscillations, resulting in increasingly pronounced variations in the effective contact dissipation. Consequently, the dissipation evolves dynamically with the contact coordinate, producing an adaptive geometric damping mechanism. Thus, the nonlinear contact potential generates nonlinear dissipation intrinsically through the contact geometry, without introducing damping phenomenologically.

\subsection{Linear Stability of the Autonomous Generalized Contact Duffing System}
For the generalized contact Duffing oscillator (\ref{E39})--(\ref{E41}), the equilibrium conditions require
\[
q^*=0,\qquad p^*=0,\qquad \Gamma(s^*)=0.
\]
For the quadratic contact potential (\ref{E27}),
\[
\Gamma(s)=\lambda s+\alpha s^2,
\]
the equilibrium contact coordinates are
\[
s^*=0,\qquad
s^*=-\frac{\lambda}{\alpha},\qquad \alpha\neq0.
\]
Thus, the nonlinear contact potential introduces an additional equilibrium that is absent in the classical linear contact formulation, enriching the contact phase structure even in the absence of external forcing.

Linearizing (\ref{E39})--(\ref{E41}) about an equilibrium point yields the Jacobian
\begin{equation}
J=
\begin{bmatrix}
0 & 1 & 0\\
-\omega_0^2-3\beta q^2 & -D(s) & -\Gamma''(s)p\\
-\omega_0^2q-\beta q^3 & p & -\Gamma'(s)
\end{bmatrix}.
\label{E43}
\end{equation}

At the classical contact equilibrium $(q^*,p^*,s^*)=(0,0,0)$,
\begin{equation}
J=
\begin{bmatrix}
0 & 1 & 0\\
-\omega_0^2 & -\lambda & 0\\
0 & 0 & -\lambda
\end{bmatrix},
\label{E44}
\end{equation}
whose eigenvalues are
\begin{equation}
\mu_1=-\lambda,\qquad
\mu_{2,3}=
\frac{-\lambda\pm\sqrt{\lambda^2-4\omega_0^2}}{2}.
\label{E45}
\end{equation}
For $\lambda>0$, all eigenvalues have negative real parts, implying that the origin is asymptotically stable.

At the additional contact equilibrium
\[
(q^*,p^*,s^*)=\left(0,0,-\frac{\lambda}{\alpha}\right),
\]
the Jacobian becomes
\begin{equation}
J=
\begin{bmatrix}
0 & 1 & 0\\
-\omega_0^2 & \lambda & 0\\
0 & 0 & \lambda
\end{bmatrix},
\label{E46}
\end{equation}
with eigenvalues
\begin{equation}
\mu_1=\lambda,\qquad
\mu_{2,3}=
\frac{\lambda\pm\sqrt{\lambda^2-4\omega_0^2}}{2}.
\label{E47}
\end{equation}
Since $\lambda>0$, at least one eigenvalue is always positive, and the equilibrium is therefore unstable.

The existence of this additional equilibrium is a direct consequence of the nonlinear contact geometry and occurs only for $\alpha\neq0$. Unlike the classical linear contact Hamiltonian, the generalized formulation modifies the contact phase structure by introducing a purely contact-induced equilibrium with fundamentally different stability properties.

\section{\label{s5}Numerical Investigations}
The dynamical consequences of the generalized contact Duffing oscillator are investigated by subjecting the mechanical subsystem to an external periodic excitation. Since the forcing acts only on the mechanical equation, the underlying contact Hamiltonian structure remains unchanged. The numerical study focuses on the quadratic contact potential (\ref{E27}), where the parameter $\alpha$ controls the strength of the nonlinear contact dissipation.

The influence of the contact geometry is examined by varying $\alpha$ while keeping all other system parameters fixed. Particular attention is given to the effective dissipation, phase-space dynamics, energy evolution, bifurcation structure, and Lyapunov stability.

For the quadratic contact potential, the governing equations are
\begin{align}
\dot q &= p,
\label{E48}\\
\dot p &= -\omega_0^2q-\beta q^3-(\lambda+2\alpha s)p+F\cos(\Omega t),
\label{E49}\\
\dot s &= \frac12p^2-\frac12\omega_0^2q^2-\frac14\beta q^4-\lambda s-\alpha s^2,
\label{E50}
\end{align}
with initial conditions
\[
q(0)=1,\qquad
p(0)=0,\qquad
s(0)=0.
\]

Unless otherwise stated, the system parameters are fixed as
\[
\omega_0=1,\quad
\beta=1,\quad
\lambda=0.2,\quad
F=0.80,\quad
\Omega=1.2,
\]
while the nonlinear contact parameter is varied over
\[
0\leq\alpha\leq 0.4.
\]

\subsection{Phase Space Dynamics}
Phase portraits of the harmonically forced generalized contact Duffing oscillator are presented in Figure~\ref{F3} for different values of the nonlinear parameter $\alpha$ to investigate the influence of the nonlinear contact potential on the asymptotic dynamics. For $\alpha=0$, corresponding to the classical contact Hamiltonian with constant geometric dissipation, the system exhibits a regular periodic attractor. As $\alpha$ increases, the phase-space structure undergoes significant changes, reflecting the state-dependent nature of the effective dissipation. The nonlinear contact geometry modifies the energy exchange between the external excitation and the intrinsic dissipation mechanism, leading to a gradual deformation and expansion of the attractor.

\begin{figure}[t]
\centering
\includegraphics[width=\linewidth]{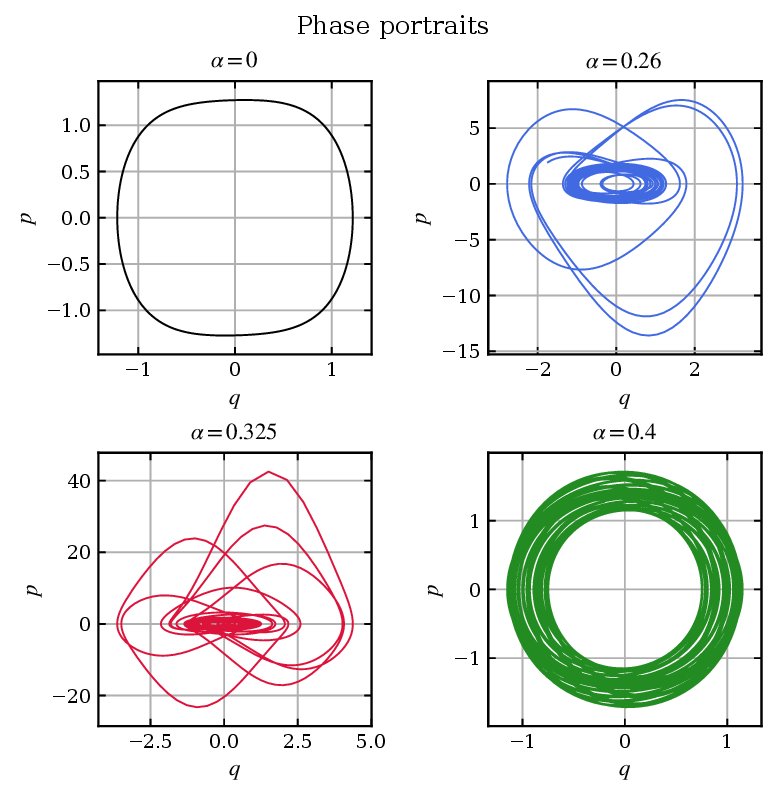}
\caption{Phase portraits of the harmonically forced generalized contact Duffing oscillator for different values of the nonlinear contact parameter $\alpha$.}
\label{F3}
\end{figure}
\subsection{Lyapunov Spectrum}

The complete Lyapunov spectrum of the generalized contact Duffing oscillator is computed to quantify its dynamical stability. Since the harmonic excitation renders the governing equations non-autonomous, the system is reformulated as an equivalent autonomous system in an extended four-dimensional phase space by introducing an additional phase variable $\phi$ satisfying $\dot{\phi}=\Omega$. The Lyapunov exponents are then computed using the standard Benettin algorithm with QR orthogonalization. The governing and variational equations are integrated simultaneously, while the tangent vectors are periodically orthonormalized to ensure numerical stability. The Lyapunov exponents are obtained from the long-time averaged logarithmic growth rates of the orthonormalized perturbation vectors. The nature of the dynamics is characterized by the largest Lyapunov exponent, which distinguishes periodic, quasiperiodic, and chaotic responses.

\begin{figure}[t]
\centering
\includegraphics[width=\linewidth]{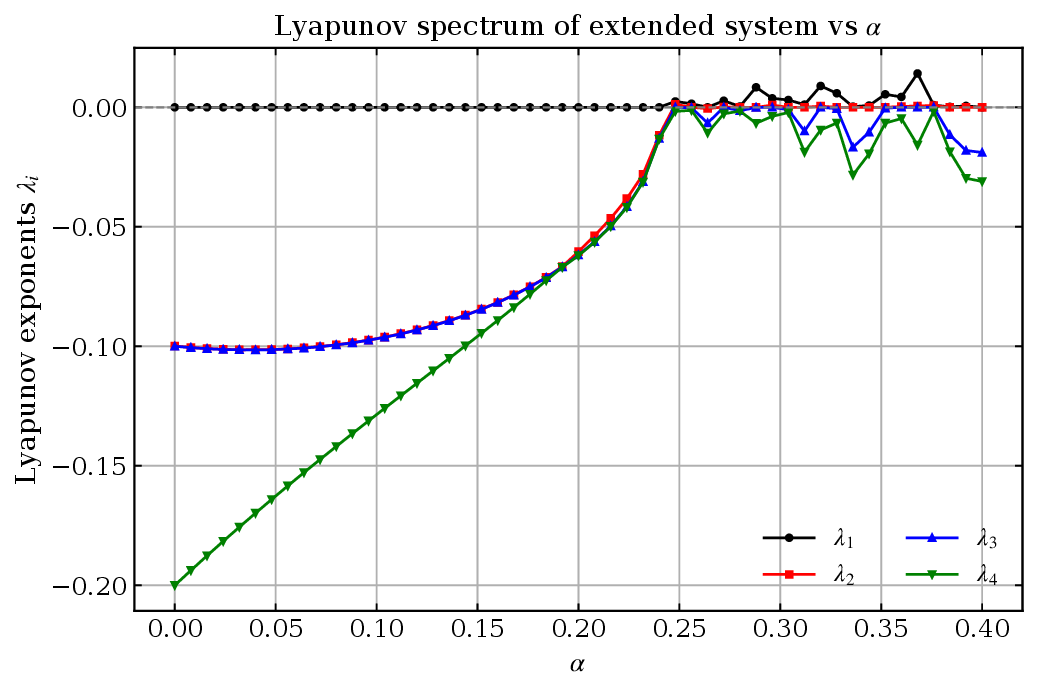}
\caption{Lyapunov spectrum of the extended four-dimensional generalized contact Duffing oscillator as a function of the nonlinear contact parameter $\alpha$. The periodically forced system is reformulated as an autonomous system in the extended phase space by introducing the forcing phase as an additional state variable. The emergence of a positive largest Lyapunov exponent indicates the transition from regular to chaotic dynamics, while the remaining exponents characterize contraction along the stable directions.}
\label{F4}
\end{figure}

Figure~\ref{F4} presents the complete Lyapunov spectrum of the generalized contact Duffing oscillator in the extended four-dimensional phase space as the nonlinear contact parameter $\alpha$ is varied. For $0\leq\alpha\leq0.245$, the largest Lyapunov exponent $\lambda_1$ (black curve) remains approximately zero, while the remaining exponents are negative, indicating bounded regular motion. Near $\alpha\approx0.245$, the spectrum undergoes a qualitative change, signaling a bifurcation and a transition in the stability of the attractor. Beyond this critical value, the largest Lyapunov exponent becomes positive, indicating sensitive dependence on initial conditions and the onset of chaotic dynamics.

As $\alpha$ is increased further, the system exhibits several quasiperiodic windows, during which the largest Lyapunov exponent returns to values close to zero. For the system parameters considered here, $\lambda_1$ again approaches zero at $\alpha=0.4$, accompanied by a second exponent that is also approximately zero, indicating quasiperiodic oscillations. These results demonstrate that the nonlinear contact parameter governs the transition from regular to chaotic dynamics through the adaptive geometric dissipation generated by the nonlinear contact potential.

\subsection{Poincar'e Sections}

Stroboscopic Poincar'e sections are constructed for the representative values of the nonlinear contact parameter $\alpha=0,;0.26,;0.325,$ and $0.4$ to characterize the asymptotic dynamics of the harmonically forced generalized contact Duffing oscillator. The governing equations are integrated over a sufficiently long time interval to eliminate transient dynamics. After discarding the transient response, the numerical solution is sampled once every forcing period,
$T=\frac{2\pi}{\Omega},$
thereby constructing the stroboscopic map
$
(q_n,p_n)=\bigl(q(nT),p(nT)\bigr).
$
The resulting Poincar'e sections provide a geometric representation of the asymptotic attractors corresponding to the selected values of $\alpha$, revealing periodic, quasiperiodic, and chaotic responses. The parameter values are chosen to coincide with those used in the phase-space analysis, allowing a direct comparison between the continuous trajectories and their corresponding stroboscopic maps. Together with the bifurcation diagram and the Lyapunov spectrum, the Poincar'e sections provide a comprehensive characterization of the global dynamics of the generalized contact Duffing oscillator.

\begin{figure}[t]
\centering
\includegraphics[width=\linewidth]{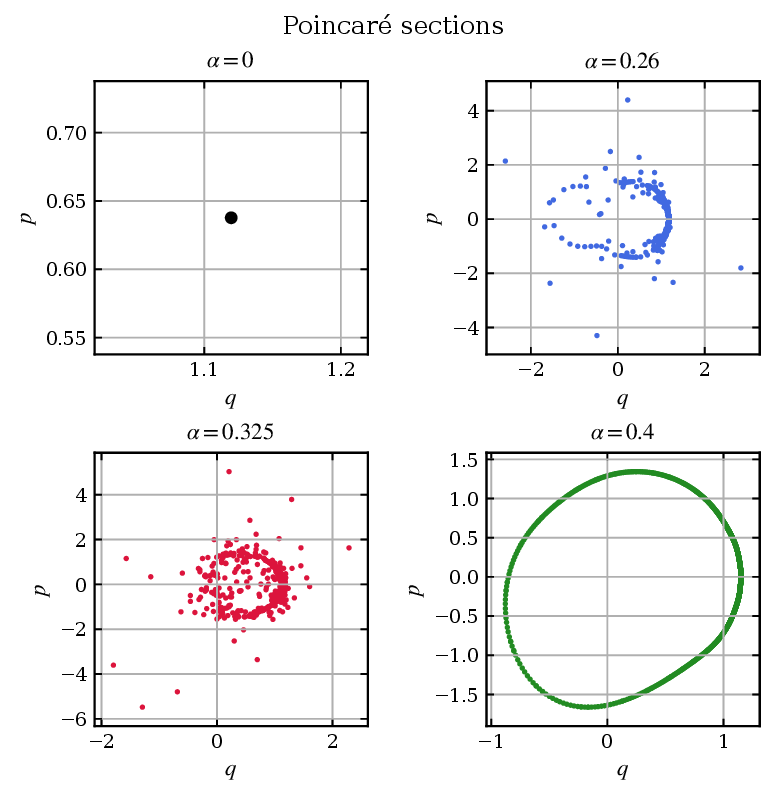}
\caption{Poincaré sections of the harmonically forced generalized contact Duffing oscillator for different values of the nonlinear contact parameter $\alpha$}
\label{F5}
\end{figure}
Figure~\ref{F5} presents the stroboscopic Poincar'e sections for the four representative values of the nonlinear contact parameter used in the phase-space analysis. For $\alpha=0$, the Poincar'e section consists of a single point, indicating a period-one limit cycle, which is consistent with the Lyapunov spectrum $(0,-,-,-)$. At $\alpha=0.26$, the section exhibits the onset of complex dynamics through the breakdown of the invariant torus, while the slightly positive largest Lyapunov exponent indicates the emergence of sensitive dependence on initial conditions. For $\alpha=0.325$, the Poincar'e section forms a scattered set of points characteristic of a chaotic attractor, in agreement with the positive largest Lyapunov exponent. Finally, at $\alpha=0.4$, the points lie on a closed invariant curve, indicating quasiperiodic motion, consistent with the Lyapunov spectrum $(0,0,-,-)$. These Poincar'e sections, together with the bifurcation diagram and Lyapunov spectrum, demonstrate how the nonlinear contact dissipation modifies the attractor geometry and governs the transition from periodic to chaotic dynamics in the generalized contact Duffing oscillator.
\section{\label{s6}Discussion}
Recent investigations have shown that contact geometry provides a natural framework for studying qualitative dynamical phenomena such as equilibria and bifurcations\cite{Montaldi2024}. The proposed nonlinear contact geometry provides a natural mechanism for generating adaptive dissipation without modifying the canonical contact Hamiltonian structure. Unlike the conventional contact formulation, the generalized framework allows the dissipation coefficient to evolve dynamically. Consequently, state-dependent damping emerges intrinsically from the contact geometry rather than being introduced phenomenologically. The generalized contact Duffing oscillator demonstrates how nonlinear contact potentials influence the global dynamics through the effective contact dissipation. Increasing the nonlinear contact parameter modifies the phase-space structure, alters the bifurcation behavior, and ultimately leads to chaotic dynamics, as confirmed by the Lyapunov spectrum. These results demonstrate that adaptive geometric dissipation provides an effective mechanism for controlling the long-term behavior of nonlinear oscillatory systems.

The proposed formulation extends the classical contact Hamiltonian framework while preserving its canonical geometric structure and enlarging the class of admissible dissipative Hamiltonian systems. The present study has been restricted to smooth contact potentials and their application to a single-degree-of-freedom Duffing oscillator. Extensions to multi-degree-of-freedom systems, non-polynomial contact potentials, and control-oriented formulations constitute promising directions for future research.
\section{\label{s7} Conclusion}
By replacing the conventional linear contact potential with an arbitrary smooth nonlinear contact potential, a generalized contact Hamiltonian framework has been developed. The proposed formulation preserves the canonical contact structure while introducing adaptive state-dependent dissipation as an intrinsic consequence of the underlying contact geometry. Within this framework, a generalized energy dissipation law and a structural characterization of contact-induced damping have been established. The proposed theory has been applied to the Duffing oscillator to investigate the dynamical implications of nonlinear contact geometry. Numerical simulations demonstrate that the generalized formulation significantly modifies the qualitative dynamical behaviour of the system through adaptive geometric dissipation. The present framework provides a systematic foundation for investigating nonlinear dissipative phenomena in a broad class of mechanical and engineering systems and opens new avenues for future research in geometric mechanics, nonlinear dynamics, and structure-preserving numerical methods.

\bibliography{reference}
\end{document}